\documentclass[11pt]{article}
\usepackage[margin=1in]{geometry}
\usepackage{amsmath,amssymb,amsthm,booktabs}
\usepackage[hidelinks]{hyperref}

\usepackage{orcidlink}

\newtheorem{theorem}{Theorem}
\newtheorem{lemma}[theorem]{Lemma}
\newtheorem{proposition}[theorem]{Proposition}
\newtheorem{claim}{Claim}[section]
\newtheorem{observation}{Observation}[section]
\newtheorem{corollary}{Corollary}[section]

\theoremstyle{remark}

\newcommand{\APS}{\operatorname{APS}}
\newcommand{\E}{\mathbb E}

\usepackage{color-edits}

\usepackage[numbers]{natbib}

\addauthor{TE}{magenta}

\addauthor{TG}{red}

\newcommand{\MMS}{\operatorname{MMS}}

\title{Fair Share Allocations for Almost All Agents}
\author{%
  \normalsize\mdseries
  Tomer Ezra\,\orcidlink{0000-0003-0626-4851}\thanks{%
    Tel Aviv University, Tel Aviv, Israel.
    Email: \texttt{tomerezra@tauex.tau.ac.il}.%
  }
  \and
  \normalsize\mdseries
  Tamar Garbuz\,\orcidlink{0009-0009-3922-4160}
}\date{}

\begin{document}
\maketitle

\begin{abstract}
We study fair allocation of indivisible goods among agents with
additive valuations. Our main result guarantees every agent her exact
AnyPrice Share (APS) whenever the total entitlement is at most
$1-\varepsilon$ and each individual entitlement is sufficiently small.
For equal entitlements, this yields a
$1$-out-of-$(n+O(\sqrt{n\log n}))$ maximin share (MMS) allocation.
We also establish two best-of-both-worlds guarantees that preserve
proportionality in expectation: each agent receives her full MMS with
probability at least $1-O((\log n/n)^{1/3})$, or every agent receives her
$1$-out-of-$(n+O(n^{2/3}(\log n)^{1/3}))$ MMS in every outcome.
Both guarantees extend to sufficiently small unequal entitlements,
with APS replacing MMS. Quantitatively, our deterministic result
allows an entitlement cap of order
$\varepsilon^2/\log(1/\varepsilon)$, while an impossibility construction
shows that any universally sufficient cap must be $O(\varepsilon)$.
\end{abstract}

\section{Introduction}

Fair division studies how to allocate resources among agents with
different preferences and claims. A natural goal for agents with equal
claims is proportionality: each of $n$ agents should receive at least a
$1/n$ fraction of her value for all the goods. When the goods are
indivisible, however, this goal may be impossible to achieve,
motivating benchmarks that account for indivisibility.

The maximin share (MMS), introduced by \citet{budish2011combinatorial},
provides one such benchmark. An agent partitions the goods into $n$
bundles and considers the least valuable bundle in her best partition.
Her MMS is the value she could secure if she chose the partition but
received its least preferred bundle. Although each agent’s MMS is defined separately, an MMS allocation must divide the goods so that every agent receives a bundle worth at least her own MMS. Such an allocation need not exist even for additive
valuations~\citep{kurokawa2018fair}.

The nonexistence of MMS allocations motivates relaxations
of the benchmark. A multiplicative relaxation guarantees each agent
a fraction of her MMS. An ordinal relaxation instead increases the
number of bundles used to define the share. For $d>n$, an agent's
$1$-out-of-$d$ MMS is the largest value she can guarantee by partitioning
the goods into $d$ bundles and receiving a least-valued bundle, while
the allocation still has only $n$ recipients. This relaxation is
ordinal because whether a bundle meets the benchmark depends only
on the agent's ranking of bundles, and is unchanged by transformations
of values that preserve this ranking~\citep{hosseini2022ordinal}.
The central question is how large $d$ must be relative to $n$
to guarantee that such an allocation exists.

Many works ask whether every additive instance admits a
$1$-out-of-$(n+1)$ MMS allocation and, more broadly, for which values of
$d$ a $1$-out-of-$d$ MMS allocation is always guaranteed
\citep{kurokawa2018fair,aigner2022envy,hosseini2021guaranteeing,
hosseini2022ordinal,akrami2023improving,babaioff2025fair,
akrami2026simultaneous}.
For four agents, \citet{schwerdtfeger20261} recently
established the $1$-out-of-$5$ guarantee, but the question remains open
for a general number of agents.

For arbitrary $n$, a sequence of works has reduced the sufficient
number of bundles. \citet{aigner2022envy} obtain $d=2n-2$ through
envy-free matchings. \citet{hosseini2021guaranteeing} subsequently obtain
$d=\lceil3n/2\rceil$,\footnote{An alternative (and equivalent) approach
asks whether every prescribed set of $n$ agents in a $d$-agent instance
can receive their original MMS~\citep{hosseini2021guaranteeing}.}
and \citet{hosseini2022ordinal} sharpen it to $d=\lfloor3n/2\rfloor$
as part of a broader theory of ordinal MMS guarantees.
\citet{akrami2023improving} subsequently establish
$d=4\lceil n/3\rceil$, improving the asymptotic bound.
Despite this progress, these general bounds require an additional
number of bundles linear in $n$. This leaves a basic asymptotic
question: can $1$-out-of-$(n+o(n))$ MMS always be guaranteed?
Such a bound would make the relative increase in the number of bundles
vanish as the population grows.

\paragraph{Unequal entitlements.}
The search for stronger ordinal guarantees also has a natural counterpart when agents have
unequal entitlements. The AnyPrice Share (APS), introduced by
\citet{babaioff2024fair}, associates an agent's entitlement with a
purchasing budget. For any nonnegative item prices summing to one,
the agent selects her most valuable affordable bundle; her APS is the
minimum value she can obtain over all such price vectors.
For entitlement $1/d$, APS is at least the $1$-out-of-$d$ MMS.
Consequently, increasing the ordinal denominator can be viewed as
introducing slack into the total entitlement: assigning entitlement
$1/d$ to each of $n$ agents gives total entitlement $n/d<1$.
This suggests a more general question. When every individual
entitlement is small, does a small amount of slack in their sum
suffice to guarantee every agent her exact APS?

\paragraph{Best-of-both-worlds fairness.}
The preceding questions concern the guarantees attainable in a single
allocation. Allowing randomization brings proportionality back into
consideration: Even when proportionality is impossible in an integral
allocation, it can always be achieved in expectation. For equal
entitlements, for example, giving all goods to a uniformly chosen agent
gives each agent her proportional share in expectation, but leaves all
other agents empty-handed in every outcome. This contrast motivates
\emph{best-of-both-worlds} fairness, which combines an \emph{ex-ante}
guarantee on expected value with an \emph{ex-post} guarantee on realized
allocations. MMS provides a natural benchmark for the latter, leading
to the question of how much of it can be guaranteed while preserving
proportionality in expectation.

For equal entitlements, \citet{babaioff2022best}
show that ex-ante proportionality can be combined with an ex-post
$1/2$-MMS guarantee. They ask whether this MMS guarantee can be
substantially improved while preserving ex-ante proportionality
\citep{babaioff2022best}.
More recently, \citet{babaioff2026near} obtain stronger guarantees
for three agents: ex-ante proportionality together with ex-post
$9/10$-MMS. Their construction also gives each agent her full MMS
with probability at least $2/3$
\citep{babaioff2026near}.
These results motivate asking how strong share-based guarantees can
become when the number of agents is large.

There are two natural ways to pursue this question.\footnote{Another
option is to guarantee every agent a constant fraction of her MMS
in every outcome, as in \citep{babaioff2022best}.}
One is to retain the original MMS and allow each agent a small
probability of falling short. The other is to require a guarantee
in every outcome, while slightly increasing the number of bundles
defining the share.
Can proportionality in expectation coexist with each agent receiving
her full MMS with probability tending to one? Can it instead coexist
with an ex-post $1$-out-of-$(n+o(n))$ MMS guarantee?
The probability in the first question is an individual guarantee;
the second requirement holds for every agent in every realized
allocation. More generally, can analogous guarantees be obtained
for agents with unequal entitlements, with APS replacing MMS?
Together with the deterministic questions above, these are the
questions we study.

\subsection{Model}\label{sec:model}

Let $N$ be a set of $n\geq 2$ agents and let $M$ be a finite set
of $m$ indivisible goods.  Every agent $i\in N$
has an entitlement $b_i\in(0,1]$ and a nonnegative additive valuation
$v_i$. Thus $v_i(g)\ge0$ is her value for a good $g\in M$, and
\[
 v_i(A)=\sum_{g\in A}v_i(g)\qquad(A\subseteq M).
\]
Write $W=\sum_{i\in N}b_i$ for the total entitlement. We consider
$W\le1$; we never rescale an agent's entitlement when defining her APS.
All logarithms are natural. 

\paragraph{Allocations.}
An \emph{integral allocation} is a family $A=(A_i)_{i\in N}$ of
pairwise disjoint bundles $A_i\subseteq M$. It is \emph{complete} if
$\bigcup_{i\in N}A_i=M$ and \emph{partial} otherwise; unless stated
otherwise, either is allowed. Nonnegativity allows us to assign leftover
goods without decreasing any agent's value.

A \emph{fractional allocation} is an array
$y=(y_{ig})_{i\in N,g\in M}$ of nonnegative numbers satisfying
\[
 \sum_{i\in N}y_{ig}\le1\qquad(g\in M).
\]
The number $y_{ig}$ is the fraction of good $g$ assigned to agent $i$.
Her fractional value is $\sum_{g\in M}y_{ig}v_i(g)$, and her
\emph{fractional support} is $\{g\in M:y_{ig}>0\}$.

A \emph{randomized allocation} is a
probability distribution with finite support over integral allocations.
For a distribution over allocations and a sampled allocation $A$, define the assignment
probability of good $g$ to agent $i$ by
\[
 \mu_{ig}=\Pr(g\in A_i).
\]
These probabilities form a fractional allocation, and additivity gives
\[
 \E[v_i(A_i)]=\sum_{g\in M}\mu_{ig}v_i(g).
\]
A guarantee is \emph{ex-ante} if it concerns expected value and \emph{ex-post} if it
holds in every realized allocation. When $W=1$, the
randomized allocation is \emph{ex-ante proportional} if
$\E[v_i(A_i)]\ge b_i v_i(M)$ for every agent $i\in N$.

\paragraph{The Maximin Share.}
For a positive integer $d$, let $\Pi_d(M)$ be the set of partitions
of $M$ into $d$ pairwise disjoint, possibly empty bundles.
Agent $i$'s \emph{$1$-out-of-$d$ maximin share} is
\[
 \MMS_i(d)
 =\max_{(B_1,\ldots,B_d)\in\Pi_d(M)}
   \min_{1\le \ell\le d}v_i(B_\ell).
\]
This is the largest value she can guarantee by partitioning the goods
into $d$ bundles and receiving a least-valued bundle.
An \emph{MMS allocation} is an integral allocation satisfying
$v_i(A_i)\ge\MMS_i(|N|)$ for every agent $i\in N$.
The usual maximin share corresponds to $d=n$. We retain the
parameter $d$ explicitly because our results also concern $d>n$.
An $\MMS(d)$ allocation satisfies $v_i(A_i)\ge\MMS_i(d)$
for every agent $i\in N$.

\paragraph{The AnyPrice Share.}
For $M\ne\varnothing$, a \emph{price vector} is a family
$\pi=(\pi_g)_{g\in M}$ of nonnegative prices with
$\sum_{g\in M}\pi_g=1$. Define the price of a bundle $A\subseteq M$
by $\pi(A)=\sum_{g\in A}\pi_g$. For an agent $i\in N$ and a budget
$b\in(0,1]$, her \emph{AnyPrice Share} is
\[
 \APS_i(b)=
 \min_{\substack{\pi_g\ge0\ (g\in M)\\\sum_{g\in M}\pi_g=1}}
 \ \max_{\substack{A\subseteq M\\\pi(A)\le b}}v_i(A).
\]
For $M=\varnothing$, set $\APS_i(b)=0$. 

For every positive integer $d$,
\begin{equation}
\label{eq:apsmms}    
 \APS_i(1/d)\ge\MMS_i(d).
\end{equation}

Indeed, under any normalized price vector, at least one bundle
in a partition attaining $\MMS_i(d)$ has price at
most $1/d$, and every bundle in that partition has value at least
$\MMS_i(d)$.
In particular, when $b_i=1/|N|$, agent $i$'s APS is at least
her usual MMS.

Define each agent's target
and the set of agents with positive targets by
\[
 t_i=\APS_i(b_i)\quad(i\in N),\qquad
 N_+=\{i\in N:t_i>0\}.
\]
An \emph{APS allocation} is an integral allocation satisfying
$v_i(A_i)\ge t_i$ for every agent $i\in N$.

\subsection{Our Contribution}

Our main result in this paper is the existence of an integral APS allocation under entitlement slack when all entitlements are small enough.
\begin{theorem}[APS with entitlement slack]\label{thm:aps}
For $0<\varepsilon\leq 1/2$, if $W\le1-\varepsilon$ and $\max_i b_i\le \frac{\varepsilon^2}{2048\log(8/\varepsilon)}$ then there exists an APS allocation.
\end{theorem}

An immediate corollary\footnote{Another interpretation of Corollary~\ref{cor:mms} is that for agents with additive valuations there always exists a $(1-\epsilon,1)$-MMS for large enough number of agents. The first parameter is the fraction of agents for which the guarantee holds, and the second parameter is what fraction of the MMS is guaranteed.} of our result combined with Inequality~\eqref{eq:apsmms} is

\begin{corollary}\label{cor:mms}
For every instance with $n$ agents, there exists an $\MMS(d)$ allocation for $$d=n+O(\sqrt{n\log(n)}).$$
\end{corollary}

\paragraph{Best of both worlds.}
We also obtain ex-ante proportional randomized allocations with two guarantees: each agent receives her APS with high probability, or every agent receives her APS at a slightly reduced entitlement in every outcome.

\begin{theorem}[Proportionality and APS with high probability]
\label{thm:bobw}
For $0<\varepsilon\leq 1/2$, if $W=1$ and $\max_i b_i\le \frac{\varepsilon^3}{262144\log(16/\varepsilon)}$ then there exists a randomized allocation $A$ guaranteeing for every agent  $i\in N$ that

\[
 \E[v_i(A_i)]\ge b_i v_i(M),\qquad
 \Pr\bigl[v_i(A_i)\ge\APS_i(b_i)\bigr]\ge1-\varepsilon.
\]
\end{theorem}

The probability in Theorem~\ref{thm:bobw} is an individual guarantee:
Each agent succeeds with probability at least $1-\varepsilon$. It does
not assert that all agents succeed simultaneously with that probability.

With respect to equal entitlements, Theorem~\ref{thm:bobw} implies:
\begin{corollary}\label{cor:mms2}
For every instance with $n$ agents, there exists a randomized allocation $A$ guaranteeing for every agent  $i\in N$ that
\[
 \E[v_i(A_i)]\ge  v_i(M)/n,\qquad
 \Pr\bigl[v_i(A_i)\ge\MMS_i(n)\bigr]\ge 1-O\left(\left(\frac{\log(n)}{n}\right)^{1/3}\right).
\] 
\end{corollary}

\begin{theorem}[Proportionality and APS at reduced entitlements]
\label{thm:bobw2}
For $0<\varepsilon\leq 1/2$, if $W=1$ and $\max_i b_i\le \frac{\varepsilon^3}{262144\log(16/\varepsilon)}$ then there exists a randomized allocation $A$ guaranteeing for every agent  $i\in N$ that

\[
 \E[v_i(A_i)]\ge b_i v_i(M),
 \qquad v_i(A_i)\ge\APS_i\!\left((1-\varepsilon)b_i\right)
 \quad\text{in every outcome}.
\]
\end{theorem}

With respect to equal entitlements, Theorem~\ref{thm:bobw2} implies:
\begin{corollary}\label{cor:mms3}
For every instance with $n$ agents, there exists a randomized allocation $A$ guaranteeing for every agent  $i\in N$ that
\[
 \E[v_i(A_i)]\ge  v_i(M)/n,\qquad
 v_i(A_i)\ge\MMS_i(d)
 \quad\text{in every outcome}
\] 
for $d=n+O\left(n^{2/3}\log^{1/3}(n)\right)$.
\end{corollary}

\paragraph{Hardness.} In Section~\ref{sec:hardness}, we complement our results and show that any entitlement cap guaranteeing
APS allocations whenever $W<1-\varepsilon$ must be $O(\varepsilon)$.

\begin{theorem}[Linear upper bound]\label{thm:hardness-cap}
Let
$c^\star(\varepsilon)$ be the supremum of the numbers $c\in[0,1]$ such
that every additive instance satisfying
\[
 \sum_{i\in N}b_i<1-\varepsilon,
 \qquad 0<b_i\le c\quad(i\in N)
\]
admits an allocation $(A_i)_{i\in N}$ with
$v_i(A_i)\ge\APS_i(b_i)$ for every agent $i\in N$. Then for every small enough  $\varepsilon> 0$ it holds that
\[
 c^\star(\varepsilon)\le 13\varepsilon.
\]
\end{theorem}

\paragraph{Polynomial implementations.}
In Section~\ref{sec:implemetaion}
we discuss how to transform our existence results into polynomial-time algorithms for finding such (randomized) allocations.

\subsection{Additional Related Work}

\paragraph{Multiplicative MMS approximation.}
The complementary objective of guaranteeing every agent a fraction
of her original MMS has been studied extensively.
\citet{kurokawa2018fair} prove that a $2/3$-MMS allocation always
exists for additive valuations. \citet{ghodsi2021fair} improve this
guarantee to $3/4$, and \citet{garg2021improved} establish the existence
of a $(3/4+1/(12n))$-MMS allocation, while also giving a strongly
polynomial-time algorithm for finding a $3/4$-MMS allocation.
\citet{akrami2024breaking} obtain the first improvement over $3/4$
by a constant independent of $n$, proving a guarantee of
$3/4+3/3836$. Subsequent work improves the guarantee to $10/13$
\citep{heidari2026improved} and then to $7/9$
\citep{huang2025fptas}. The latter work also gives an algorithm
computing a $(7/9-\varepsilon)$-MMS allocation in time polynomial
in the input size and $1/\varepsilon$.

On the impossibility side, \citet{feige2021tight} construct an
instance with three agents and nine goods in which every allocation
gives some agent at most a $39/40$ fraction of her MMS.
For every $n\ge4$, they also construct $n$-agent instances with
an upper bound of $1-1/n^4$ on the simultaneously attainable
fraction of MMS. \citet{ezra2026improved} strengthen the constant
impossibility bound to $20/21$ using an instance with four agents
and eleven goods. They also strengthen the asymptotic bound to
$1-\Omega((\log n)^{-2})$ for every sufficiently large $n$.

Ordinal guarantees concern a different approximation parameter.
For example, if an agent values each of $n$ goods at one, then
$\mathrm{MMS}_i(n)=1$, whereas $\mathrm{MMS}_i(n+1)=0$, since
every partition into $n+1$ bundles contains an empty bundle.
Thus, even an ordinal denominator of $n+1$ does not by itself
imply a positive multiplicative approximation to the original MMS.

\paragraph{Unequal entitlements and other best-of-both-worlds guarantees.}
\citet{farhadi2019fair} introduce weighted MMS for agents with
unequal entitlements and establish a tight $1/n$ approximation
guarantee for additive valuations.
\citet{babaioff2024fair} introduce APS and give a polynomial-time
algorithm guaranteeing every agent at least $3/5$ of her APS
for arbitrary entitlements.
\citet{chakraborty2024weighted} study relaxations of weighted
envy-freeness and proportionality and their relationships with
share-based benchmarks; \citet{suksompong2025weighted} surveys
the broader literature on weighted fair division.
More generally, \citet{babaioff2025fair} study which share
benchmarks are always simultaneously attainable and how such
benchmarks compare. For arbitrary entitlements,
\citet{babaioff2025share} show that every agent can simultaneously
receive at least half the value prescribed by any feasible share
benchmark she chooses, even when different agents choose different
benchmarks.

Alongside the share-based best-of-both-worlds results discussed
above, another line of work combines envy-based guarantees before
and after randomization.
For additive valuations and equal entitlements,
\citet{aziz2024best} give a polynomial-time algorithm computing a
randomized allocation that is ex-ante envy-free and ex-post envy-free up to one
good (EF1). Thus, no agent prefers another agent's allocation in
expectation, and in every outcome any envy can be eliminated by
removing one good from the envied bundle.
For unequal entitlements, \citet{hoefer2024best} give a strongly
polynomial-time algorithm computing a randomized allocation that is ex-ante
weighted envy-free and simultaneously satisfies two ex-post
guarantees: weighted proportionality up to one good (WPROP1)
and weighted transfer envy-freeness up to one good
($\mathrm{WEF}(1,1)$).
The first means that each agent can reach her proportional share
by adding at most one good to her bundle. The second means that
any weighted envy can be eliminated by hypothetically transferring
one good from the envied agent to the envying agent.

\paragraph{Proof ingredients and broader background.}
We use the reduction to ordered instances of
\citet{bouveret2016characterizing}, in the formulation of
\citet{barman2020approximation}, and the faithful
rounding guarantee of
\citet{babaioff2022best}.
Our joint-sampling lemma uses the linear-dependence principle
underlying discrepancy rounding
\citep{beck1981integer,doerr2007roundings}.
For a broader overview of fair division of indivisible goods and
its open questions, we refer to \citet{amanatidis2023fair}.

\subsection{Preliminaries}
\label{sec:prelim}
\paragraph{Ordered goods.}
We use the standard reduction to ordered additive instances
\citep{barman2020approximation}.
For an original instance, let $[m]=\{1,\ldots,m\}$ be the set of ranks,
and let $v_i^{\rm ord}(r)$ be agent $i$'s $r$-th largest singleton
value. These values define her additive ordered valuation on $[m]$.
The reduction preserves her APS and total value; in particular,
$v_i^{\rm ord}([m])=v_i(M)$. Each allocation of ranks
$B=(B_i)_{i\in N}$ can be converted to an original allocation
$A=(A_i)_{i\in N}$ satisfying
\[
 v_i(A_i)\ge v_i^{\rm ord}(B_i)\qquad(i\in N).
\]
Applying this pointwise guarantee to each lottery outcome also preserves
the expected-value and individual-probability guarantees above.
Throughout the existence proofs, we therefore identify goods with
ranks, write $v_i$ for the ordered valuation, and assume
\[
 M=[m],\qquad v_i(1)\ge\cdots\ge v_i(m)\ge0\quad(i\in N).
\]

\paragraph{APS Certificates.} A distribution $D$ on bundles assigns a number $D(P)\ge0$ to each
$P\subseteq M$, with $\sum_{P\subseteq M}D(P)=1$. Its support is
$\{P\subseteq M:D(P)>0\}$, and $P\sim D$ means that $P$ is sampled
from this distribution.

\begin{lemma}[APS Certificates {\citep{babaioff2024fair}}]
\label{lem:certificate}
For every agent $i\in N$ and every number $t\ge0$, we have
$t\le\APS_i(b_i)$ if and only if there is a distribution $D$ on
bundles satisfying
\[
 v_i(P)\ge t\quad(P\subseteq M:D(P)>0),\qquad
 \Pr_{P\sim D}(g\in P)\le b_i\quad(g\in M).
\]
Such a distribution is called an \emph{APS certificate} for $t$.
\end{lemma}

The next lemma applies to any finite agent set $J$, any finite good set
$U$, and any nonnegative additive valuations $(u_i)_{i\in J}$.

\begin{lemma}[Faithful Rounding {\citep{babaioff2022best}}]
\label{lem:round}
Let $y=(y_{ig})_{i\in J,g\in U}$ be a fractional allocation. For
every agent $i\in J$, let $m_i\ge0$ satisfy
$u_i(g)\le m_i$ whenever $y_{ig}>0$. There is a randomized allocation   $(B_i)_{i\in J}$ such that
\[
 \Pr(g\in B_i)=y_{ig}\quad(i\in J,\ g\in U)
\]
and, in every outcome, each agent $i\in J$ satisfies
\[
 B_i\subseteq\{g\in U:y_{ig}>0\},\qquad
 u_i(B_i)\ge\sum_{g\in U}y_{ig}u_i(g)-m_i.
\]
\end{lemma}

\section{Building Allocations from APS Certificates}
\label{sec:ingredients}

Fix an integer $q\ge16$. Throughout this section, we assume that
 $b_i\le1/q^2$ for every $i\in N$, and that the set
$N_+=\{i\in N:t_i>0\}$ of active agents is nonempty (otherwise, Theorems~\ref{thm:aps}-\ref{thm:bobw2} are immediate).
All configurations and requests (defined later) in this section concern
active agents.

\subsection{Capped Values and the Partitions of Items}

For every active agent $i\in N_+$, define her capped additive valuation by
\[
 u_i(g)=\min\{v_i(g)/t_i,1\}\quad(g\in M),\qquad
 u_i(B)=\sum_{g\in B}u_i(g)\quad(B\subseteq M).
\]
Choose an APS certificate $D_i$ for $t_i$ using
Lemma~\ref{lem:certificate}. Thus $D_i$ is a distribution on bundles
$P\subseteq M$ with $v_i(P)\ge t_i$ throughout its support and
$\Pr_{P\sim D_i}(g\in P)\le b_i$ for every good $g$.

\begin{claim}[Capping preserves the target]
\label{cl:capping}
For every active agent $i$, the certificate $D_i$ satisfies
$u_i(P)\ge1$ throughout its support, and $u_i(M)\ge1/b_i$.
\end{claim}
\begin{proof}
Fix an active agent $i\in N_+$ and a bundle $P$ in the support of $D_i$.
If some good $g\in P$ has $v_i(g)\ge t_i$,
then $u_i(g)=1$. Otherwise no value in $P$ is capped, and
$u_i(P)=v_i(P)/t_i\ge1$. Therefore
\[
 1\le\E_{P\sim D_i}[u_i(P)]
 =\sum_{g\in M}u_i(g)\Pr_{P\sim D_i}(g\in P)
 \le b_i u_i(M),
\]
which concludes the proof of the claim.
\end{proof}

The capped values remain nonincreasing in the order of the goods.
They are at most one, so Claim~\ref{cl:capping} implies
$m\ge u_i(M)\ge1/b_i\ge q^2>q$.

For every positive integer $h$, define the prescribed bucket size
\[
 w_h=\left\lceil q(1+1/q)^{h-1}\right\rceil.
\]
Let $H$ be the least positive integer with $\sum_{h=1}^H w_h\ge m$.
We create two parallel partitions of the goods: (1) consecutive buckets
$C_1,\ldots,C_H$, and (2) reserve goods $R$, buffer goods $E$, and ordinary
goods $S$. The buckets have the prescribed size (i.e., for every $h<H$, $|C_h|=w_h$), except that the last
is shortened if necessary. Explicitly, for $1\le h\le H$, set
\[
 C_h=\left\{1+\sum_{k=1}^{h-1}w_k,\ldots,
              \min\!\left\{m,\sum_{k=1}^h w_k\right\}\right\}.
\]
Define the reserve goods by
\[
 R=\{1,1+q,1+2q,\ldots\}\cap M,
 \] and the buffer goods in bucket $C_h$ by \[
 E_h=\{\text{the first $\lfloor|C_h|/q\rfloor$ goods of $C_h\setminus R$}\} .
\]
Define the complete buffer set, the ordinary set, and its bucket portions by
\[
 E=\bigcup_{h=1}^H E_h,\qquad
 S=M\setminus(R\cup E),\qquad S_h=S\cap C_h.
\]
Since $m>q$ and $|C_1|=q$, we have $H\ge2$. 

\begin{claim}
\label{cl:geometry}
The buffer sets are well-defined. Moreover,
\[
 |C_j|\le(1+2/q)|C_{j-1}|\quad(2\le j\le H),\qquad
 |S_h|\ge(1-3/q)|C_h|\quad(1\le h<H).
\]
\end{claim}
\begin{proof}
Fix a bucket $C_h$ and write $k=|C_h|$. If $k<q$, its buffer set
is empty. If $k\ge q$, it contains at most $k/q+1$ reserve goods, so
\[
 |C_h\setminus R|\ge k-k/q-1\ge k-2k/q\ge k/q.
\]
Thus there are enough goods to define $E_h$. For a full bucket,
\[
 |S_h|\ge k-(k/q+1)-k/q\ge(1-3/q)k.
\]
For every $2\le j\le H$, the preceding bucket is full. Hence
\[
 |C_j|\le w_j\le(1+1/q)w_{j-1}+1
 \le(1+2/q)|C_{j-1}|,
\]
which concludes the proof of the claim.
\end{proof}

\subsection{Cutoffs, Configurations, and Requests}

For every active agent $i$, define her cutoff index and tail by
\[
 h_i=\min\bigl(\{h\in\{1,\ldots,H-1\}:|C_h|\ge q/b_i\}
                   \cup\{H\}\bigr),\qquad
 T_i=\bigcup_{j=h_i+1}^H C_j.
\]
Her head is $M\setminus T_i$. Thus $h_i$ is the first full bucket
reaching size $q/b_i$, or the last bucket if there is no such full
bucket. When $h_i=H$, the tail is empty.

\begin{claim}[Enough buffers at every cutoff]
\label{cl:buffers}
For every $h\in\{1,\ldots,H-1\}$,
\[
 |\{i\in N_+:h_i=h\}|\le\left\lfloor\frac{|C_h|}{q}\right\rfloor
 =|E_h|.
\]
\end{claim}
\begin{proof}
Every agent counted on the left has $b_i\ge q/|C_h|$. Therefore
\[
 |\{i\in N_+:h_i=h\}|\frac q{|C_h|}
 \le\sum_{i\in N_+:h_i=h}b_i\le W\le1.
\]
Taking the integer part gives the claim.
\end{proof}

\begin{claim}[Cutoff bound]
\label{cl:cutoff}
Every active agent $i$ satisfies $h_i\le2+2q\log(1/b_i)$.
\end{claim}
\begin{proof}
Define the prescribed bucket index
\[
 k_i=1+\left\lceil\frac{\log(1/b_i)}{\log(1+1/q)}\right\rceil.
\]
Its prescribed size satisfies $w_{k_i}\ge q/b_i$. If $k_i<H$, the
cutoff is at most $k_i$; otherwise $h_i\le H\le k_i$. Since
$\log(1+1/q)\ge1/(2q)$, in either case
\[
 h_i\le k_i\le2+\frac{\log(1/b_i)}{\log(1+1/q)}
 \le2+2q\log(1/b_i),
\]
which concludes the proof of the claim.
\end{proof}

For every active agent $i$ and every bundle $P\subseteq M$, define
the proportion of her tail value contained in $P$ by
\[
 \lambda_{iP}=\begin{cases}
 u_i(P\cap T_i)/u_i(T_i),&u_i(T_i)>0,\\
 0,&u_i(T_i)=0.
 \end{cases}
\]
It holds  that
$0\le\lambda_{iP}\le1$ and
$\lambda_{iP}u_i(T_i)=u_i(P\cap T_i)$, including when the tail has
value zero.

For every index $j\in\{2,\ldots,H\}$, active agent $i$, and
bundle $P\subseteq M$, define the \textit{request}  by
\[
 d_{j,iP}=\begin{cases}
 |P\cap C_j|,&j\le h_i,\\
 |C_j|\lambda_{iP},&j>h_i.
 \end{cases}
\]

A \emph{configuration} for agent $i$ is a bundle $P\subseteq M$
with $u_i(P)\ge1$ and $P\cap C_1=\varnothing$.

\begin{claim}[Expected requests from marginal bounds]
\label{cl:request-means}
Fix an active agent $i$, a number $a_i\ge0$, and a distribution $Q_i$
on bundles $P\subseteq M$ such that
$\Pr_{P\sim Q_i}(g\in P)\le a_i$ for every good $g$.
Then
\[
 \E_{P\sim Q_i}[\lambda_{iP}]\le a_i,\qquad
 \E_{P\sim Q_i}[d_{j,iP}]\le a_i|C_j|
 \quad(2\le j\le H).
\]
\end{claim}
\begin{proof}
If $u_i(T_i)>0$, the expected tail proportion satisfies
\[
 \E_{P\sim Q_i}[\lambda_{iP}]
 =\frac{\sum_{g\in T_i}u_i(g)\Pr_{P\sim Q_i}(g\in P)}{u_i(T_i)}
 \le a_i.
\]
If $u_i(T_i)=0$, the proportion is identically zero. For every index $j\in \{2,\ldots,h_i\}$, the expected request is
\[
 \E_{P\sim Q_i}[d_{j,iP}]
 =\sum_{g\in C_j}\Pr_{P\sim Q_i}(g\in P)\le a_i|C_j|.
\]
For every tail index  $j>h_i$, the expected request is
\[
 \E_{P\sim Q_i}[d_{j,iP}]
 =|C_j|\E_{P\sim Q_i}[\lambda_{iP}]\le a_i|C_j|,
\]
which concludes the proof of the claim.
\end{proof}

\subsection{Realizing Selected Configurations}

\begin{lemma}
\label{lem:realize}
Let $I\subseteq N_+$ be a subset of agents, and let $\{P_i\}_{i \in I}$ be a collection of configurations  satisfying:
\begin{equation}\label{eq:realize-capacity}
 \sum_{i\in I}d_{j,iP_i}\le|S_{j-1}|\qquad(2\le j\le H).
\end{equation}
There is a randomized partial allocation $(B_i)_{i\in I}$
using only items $S\cup E$, such that
\[
 u_i(B_i)\ge u_i(P_i)\quad(i\in I)
\]
in every outcome. 
For every agent $i\in I$, her assignment probabilities are
\[
\Pr(g\in B_i)=
\begin{cases}
 \dfrac{d_{j,iP_i}}{|S_{j-1}|},
 & g\in S_{j-1},\quad 2\le j\le H,\\[6pt]
 \dfrac{1}{|E_{h_i}|},
 & g\in E_{h_i},\quad 1\le h_i<H,\\[6pt]
 0,
 & \text{otherwise}.
\end{cases}
\]
All goods in $R$, $S_H$, and $E_H$ remain unused.
\end{lemma}
\begin{proof}
For every selected agent $i\in I$, put $\lambda_i=\lambda_{iP_i}$.

\paragraph{Head assignments.}
For every index $j\in\{2,\ldots,H\}$, create one head request for
agent $i$ for each good in $P_i\cap C_j$, whenever $i\in I$ and
$j\le h_i$. Define the number of these requests by
\[
 n_j=\sum_{i\in I:j\le h_i}|P_i\cap C_j|.
\]
Capacity gives $n_j\le|S_{j-1}|$. Assign the requests by a uniformly
random injection into $S_{j-1}$, that is, choose uniformly among all
assignments of distinct supply goods to the requests.
Let $B_i^{\rm head}$ be the union
of agent $i$'s replacements. Every replacement precedes its requested
good, and $P_i$ avoids $C_1$, so
\begin{equation}\label{eq:realize-head}
 u_i(B_i^{\rm head})\ge\sum_{j=2}^{h_i}u_i(P_i\cap C_j)
 =u_i(P_i\setminus T_i).
\end{equation}

\paragraph{Unused supplies and fractional tails.}
For every index $j\in\{2,\ldots,H\}$, define its unused goods and
their number, and then their union, by
\[
 U_j=S_{j-1}\setminus\bigcup_{i\in I}B_i^{\rm head},\qquad
 r_j=|U_j|=|S_{j-1}|-n_j,\qquad U=\bigcup_{j=2}^H U_j.
\]
Subtracting the head requests from \eqref{eq:realize-capacity} gives
\begin{equation}\label{eq:realize-residual}
 |C_j|\sum_{i\in I:h_i<j}\lambda_i\le r_j.
\end{equation}
For every selected agent $i$ and good $g\in U_j$, define its
fractional tail share by
\[
 y_{ig}=\begin{cases}|C_j|\lambda_i/r_j,&h_i<j,\\
 0,&h_i\ge j.
 \end{cases}
\]
The denominator is positive whenever such a good $g$ exists.
For every good $g\in U_j$, Equation~\eqref{eq:realize-residual} gives
\[
 \sum_{i\in I}y_{ig}
 =\frac{|C_j|}{r_j}\sum_{i\in I:h_i<j}\lambda_i\le1.
\]

Fix a selected agent $i$ and an index $j>h_i$ with $r_j>0$.
Every good in $U_j\subseteq C_{j-1}$ precedes every good in $C_j$,
so
\[
 \sum_{g\in U_j}y_{ig}u_i(g)
 =\frac{|C_j|\lambda_i}{r_j}u_i(U_j)
 \ge\frac{|C_j|\lambda_i}{r_j}
        \frac{r_j u_i(C_j)}{|C_j|}
 =\lambda_i u_i(C_j).
\]
If $r_j=0$, \eqref{eq:realize-residual} forces $\lambda_i=0$ for
every $i\in I$ with $h_i<j$, so this bound also holds with both
sides zero. Summing gives
\begin{equation}\label{eq:realize-tail}
 \sum_{g\in U}y_{ig}u_i(g)\ge\lambda_i u_i(T_i)=u_i(P_i\cap T_i).
\end{equation}

\paragraph{Buffers and rounding.}
For every $h\in\{1,\ldots,H-1\}$, define
$I_h=\{i\in I:h_i=h\}$. By Claim~\ref{cl:buffers}, we can choose
a uniformly random injection of $I_h$ into $E_h$. Write $e_i$ for
the buffer item given to an agent with $h_i<H$.
For each such agent, every good in her positive fractional tail support (defined by $y_{ig}$)
lies in $S_{h_i}$ or a later supply. Since $E_{h_i}$ consists of the
first nonreserve goods of that bucket, her buffer precedes that support:
\[
 y_{ig}>0\quad\Longrightarrow\quad u_i(g)\le u_i(e_i).
\]
Apply Lemma~\ref{lem:round} to the fractional allocation $y$, using
loss bound $u_i(e_i)$ when $h_i<H$ and zero when $h_i=H$ (where in this case, $y_{ig}=0$ for every $g\in U$).
Condition on the head and buffer assignments when applying that lemma.
Let $B_i^{\rm tail}$ be its rounded bundles. They preserve every
conditional fraction $y_{ig}$. Define
\[
 B_i=\begin{cases}
 B_i^{\rm head}\cup B_i^{\rm tail}\cup\{e_i\},&i\in I,\ h_i<H,\\
 B_i^{\rm head},&i\in I,\ h_i=H,\\
 \varnothing,&i\notin I.
 \end{cases}
\]
For $i\in I$ with $h_i<H$, the rounding bound and Equations~
\eqref{eq:realize-head} and \eqref{eq:realize-tail} give
\[
 u_i(B_i)\ge u_i(P_i\setminus T_i)
             +u_i(P_i\cap T_i)-u_i(e_i)+u_i(e_i)=u_i(P_i).
\]
If $h_i=H$, the tail shares vanish and the head alone gives this bound.
All bundles are disjoint: heads use distinct goods, tails use unused
ordinary goods, and buffers use a separate supply.

\paragraph{Assignment probabilities.}
Fix $i\in I$ and $g\in S_{j-1}$ for an index $2\le j\le H$.
If $j\le h_i$, uniform head assignment gives
\[
 \Pr(g\in B_i)=\frac{|P_i\cap C_j|}{|S_{j-1}|}
 =\frac{d_{j,iP_i}}{|S_{j-1}|}.
\]
If $h_i<j$ and $r_j>0$, uniform head assignment leaves $g$ unused
with probability $r_j/|S_{j-1}|$. Conditional on being unused, its
fraction and its rounding probability are $|C_j|\lambda_i/r_j$.
Therefore
\[
 \Pr(g\in B_i)=\frac{r_j}{|S_{j-1}|}\frac{|C_j|\lambda_i}{r_j}
 =\frac{d_{j,iP_i}}{|S_{j-1}|}.
\]
If $r_j=0$, the tail request is zero by
\eqref{eq:realize-residual}. The buffer probabilities follow from
their uniform injections. No step uses $R$, $S_H$, or $E_H$.
\end{proof}

\subsection{Joint Sampling with Bounded Error}
Let $I$ be a finite nonempty set of agents. For each $i\in I$, let
$\mathcal P_i$ be a finite nonempty set of choices and let $Q_i$ be a
probability distribution on $\mathcal P_i$.
Fix a positive integer $s$ and a real number $\xi\ge0$.
For every $i\in I$ and $r\in\{1,\ldots,s\}$, let
$f_{r,i}:\mathcal P_i\to\mathbb R$ be a function.
For a tuple of choices $(P_i)_{i\in I}$, the $r$th sum that we wish
to control is $\sum_{i\in I}f_{r,i}(P_i)$.
The values $f_{r,i}(P)$ may be positive or negative. Assume that
\[
 \sum_{r=1}^s |f_{r,i}(P)|\le\xi
 \qquad(i\in I,\ P\in\mathcal P_i).
\]

\begin{lemma}[Joint sampling with bounded error]
\label{lem:doerr-distribution-preserving}
There is a joint distribution of choices $(P_i)_{i\in I}$ such that
\[
 \Pr(P_i=P)=Q_i(P)
 \qquad(i\in I,\ P\in\mathcal P_i),
\]
and, in every outcome,
\[
 \left|\sum_{i\in I}f_{r,i}(P_i)
       -\sum_{i\in I}\mathbb E_{P\sim Q_i}[f_{r,i}(P)]\right|
 \le2\xi
 \qquad(1\le r\le s).
\]
The choices may be dependent.
\end{lemma}

The proof constructs the joint distribution by induction on the number
of probabilities strictly between zero and one. At each step, we express
the current probabilities as a convex combination of two collections
with fewer fractional probabilities, apply induction to each, and mix
the resulting joint distributions. We preserve the expected sums that
could otherwise change by more than $2\xi$; all remaining sums satisfy
the error bound directly.

\begin{proof}
If $\xi=0$, every $f_{r,i}(P)$ is zero, so independent choices suffice.
Assume $\xi>0$. Write $x_{iP}=Q_i(P)$ and let
\[
 F=\{(i,P):i\in I,\ P\in\mathcal P_i,\ 0<x_{iP}<1\},
 \qquad k=|F|.
\]
We prove the claim for every collection of prescribed distributions by
induction on $k$. If $k=0$, each agent has a unique choice of probability
one. Choosing these choices deterministically gives zero error in every
sum. Suppose henceforth that $k>0$.

\paragraph{Splitting the probabilities.}
Define the set of indices whose expected sums we will preserve by
\[
 J=\left\{r\in\{1,\ldots,s\}:
       \sum_{(i,P)\in F}|f_{r,i}(P)|>2\xi\right\}.
\]
We seek real numbers $d_{iP}$, not all zero, with $d_{iP}=0$ outside
$F$, satisfying the homogeneous linear equations
\begin{align}
 \sum_{P\in\mathcal P_i}d_{iP}&=0
 &&(i\in I),
 \label{eq:joint-sampling-agent-sums}\\
 \sum_{(i,P)\in F}f_{r,i}(P)d_{iP}&=0
 &&(r\in J).
 \label{eq:joint-sampling-preserved-sums}
\end{align}
Equation~\eqref{eq:joint-sampling-agent-sums} ensures that replacing
$x$ by $x+td$ preserves each agent's total probability of one.
Equation~\eqref{eq:joint-sampling-preserved-sums} ensures that
\[
 \sum_{i\in I}\sum_{P\in\mathcal P_i}
       (x_{iP}+td_{iP})f_{r,i}(P)
 =\sum_{i\in I}\sum_{P\in\mathcal P_i}x_{iP}f_{r,i}(P)
 \qquad(r\in J).
\]
Thus, whenever $x+td$ has entries in $[0,1]$, it prescribes probability
distributions with the same expected sums for all $r\in J$.

To see that a nonzero solution $d$ exists, first observe that
\[
 \sum_{r=1}^s\sum_{(i,P)\in F}|f_{r,i}(P)|
 =\sum_{(i,P)\in F}\sum_{r=1}^s|f_{r,i}(P)|
 \le k\xi.
\]
The inequality follows by applying the hypothesis to each of the $k$
pairs in $F$. By the definition of $J$, it implies $|J|<k/2$.
Moreover, every agent with a fractional probability has at least two such probabilities:
a single number strictly between zero and one cannot sum to one with
numbers in $\{0,1\}$. Consequently, at most $k/2$ agents have fractional
probabilities. Only these agents give nontrivial equations in
\eqref{eq:joint-sampling-agent-sums}. Together with
\eqref{eq:joint-sampling-preserved-sums}, there are therefore fewer
than $k$ nontrivial homogeneous equations in the $k$ unknowns
$(d_{iP})_{(i,P)\in F}$. Hence a nonzero solution exists.

Let $t_+>0$ be the largest number for which all entries of $x+t_+d$
belong to $[0,1]$, and let $t_->0$ be the largest number for which all
entries of $x-t_-d$ belong to $[0,1]$. Both numbers are positive because
every entry that can change is initially strictly between zero and one,
and both are finite because $d\ne0$. Set
\[
 x^+=x+t_+d,\qquad x^-=x-t_-d,
 \qquad \theta=\frac{t_-}{t_++t_-}.
\]
The agent equations ensure that $x^+$ and $x^-$ each prescribe a
probability distribution for every agent. At each endpoint, at least
one previously fractional entry reaches zero or one; otherwise the
corresponding step could be increased. Entries outside $F$ do not
change. Thus both endpoints have fewer than $k$ fractional entries, and
\[
 x=\theta x^++(1-\theta)x^-,
 \qquad 0<\theta<1.
\]

\paragraph{Constructing the joint distribution.}
Apply the induction hypothesis to $x^+$ and $x^-$, obtaining joint
distributions $\mathcal D^+$ and $\mathcal D^-$ with those marginals.
Let $\mathcal D$ sample from $\mathcal D^+$ with probability $\theta$
and from $\mathcal D^-$ with probability $1-\theta$.
For every $i$ and $P$,
\[
 \Pr_{\mathcal D}(P_i=P)
 =\theta x^+_{iP}+(1-\theta)x^-_{iP}
 =x_{iP}=Q_i(P).
\]
The first equality uses the marginal guarantees of the two distributions,
and the second uses the convex combination above.

\paragraph{Verifying the error bound.}
Fix $r\in\{1,\ldots,s\}$. If $r\in J$, then
\[
 \sum_{i\in I}\sum_{P\in\mathcal P_i}x^\pm_{iP}f_{r,i}(P)
 =\sum_{i\in I}\sum_{P\in\mathcal P_i}x_{iP}f_{r,i}(P).
\]
This equality follows from
\eqref{eq:joint-sampling-preserved-sums} and holds separately for
$x^+$ and $x^-$.
Induction bounds the error in every outcome of each component
distribution by $2\xi$ relative to its own expected sum.
Since both expected sums equal the original one, the same bound holds
in every outcome of $\mathcal D$.

If $r\notin J$, fix an outcome $(P_i)_{i\in I}$ in the support of
$\mathcal D$. For every $(i,P)\notin F$, we have
$\mathbf 1_{\{P_i=P\}}=x_{iP}$: the marginal probability $x_{iP}$ is
zero or one and therefore forces the corresponding indicator.
Consequently,
\begin{align*}
 &\left|\sum_{i\in I}f_{r,i}(P_i)
       -\sum_{i\in I}\sum_{P\in\mathcal P_i}x_{iP}f_{r,i}(P)\right|\\
 &\qquad=
 \left|\sum_{(i,P)\in F}
       \bigl(\mathbf 1_{\{P_i=P\}}-x_{iP}\bigr)f_{r,i}(P)\right|\\
 &\qquad\le\sum_{(i,P)\in F}|f_{r,i}(P)|
 \le2\xi.
\end{align*}
The equality cancels the indicators outside $F$. The first inequality
uses the triangle inequality and
$|\mathbf 1_{\{P_i=P\}}-x_{iP}|\le1$; the second uses $r\notin J$.
Thus every sum satisfies the required bound. Each recursive step
reduces $k$, so the construction terminates after finitely many steps.
\end{proof}

The decomposition uses the same linear-dependence argument as
Doerr's rounding proof~\cite[Theorem~6(b)]{doerr2007roundings}.
The proof above is self-contained and explicitly constructs a joint
distribution with the prescribed marginals.
\section{An APS Allocation with Slack}
\label{sec:aps}

We prove Theorem~\ref{thm:aps}. Fix an instance satisfying its
assumptions, and set
\[
 q=\frac8\varepsilon,
 \qquad c=c_\varepsilon=\frac1{32q^2\log q}.
\]
Thus $q\ge16$, $W\le1-8/q$, and $b_i\le c$ for every agent
$i\in N$. In particular,
\begin{equation}\label{eq:aps-parameters}
 c\le\frac1{q^2},\qquad qc\le\frac1q.
\end{equation}
If $N_+=\varnothing$, any allocation proves the theorem. Otherwise,
use the normalized valuations, certificates, and bucket construction
from Section~\ref{sec:ingredients} with this value of $q$.
We first select configurations for most active agents, and then use
the reserve to serve the remaining active agents.

\subsection{A Feasible System of Configuration Weights}

For every active agent $i\in N_+$, let $\widetilde D_i$ be her
certificate distribution $D_i$ conditioned on avoiding the first
bucket $C_1$. Let $\mathcal P_i$ denote its support: the configurations
with positive probability under $\widetilde D_i$. This conditioning
is well-defined because $|C_1|=q$ and the certificate has item
marginals at most $b_i$, so
\[
 \Pr_{P\sim D_i}(P\cap C_1\ne\varnothing)
 \le\sum_{g\in C_1}\Pr_{P\sim D_i}(g\in P)
 \le qb_i\le qc<1.
\]
Every configuration $P\in\mathcal P_i$ satisfies $u_i(P)\ge1$ and
$P\cap C_1=\varnothing$. For every good $g\in M$, its conditional
inclusion probability satisfies
\begin{equation}\label{eq:aps-conditional}
 \Pr_{P\sim\widetilde D_i}(g\in P)
 \le\frac{b_i}{1-qb_i}\le\frac{b_i}{1-qc}.
\end{equation}

For every agent $i\in N_+$ and configuration $P\in\mathcal P_i$,
introduce a nonnegative weight $x_{iP}$. Using the request
coefficients $d_{j,iP}$ defined in Section~\ref{sec:ingredients},
consider the system
\begin{align}
 \sum_{P\in\mathcal P_i}x_{iP}&=1
 &&(i\in N_+),\label{eq:aps-normalization}\\
 \sum_{i\in N_+}\sum_{P\in\mathcal P_i}d_{j,iP}x_{iP}
 &\le |S_{j-1}|
 &&(2\le j\le H).\label{eq:aps-capacity}
\end{align}
The first constraints give each agent total configuration weight
one. The constraint indexed by $j$ bounds the total request from
the ordinary goods in $S_{j-1}$; we call it capacity constraint $j$.

\begin{claim}[Feasibility]\label{cl:aps-feasible}
The weights $x_{iP}=\widetilde D_i(P)$ satisfy
\eqref{eq:aps-normalization}--\eqref{eq:aps-capacity}.
\end{claim}

\begin{proof}
Normalization follows because every $\widetilde D_i$ is a
probability distribution. By Claim~\ref{cl:request-means} and
\eqref{eq:aps-conditional}, the expected request in each capacity
row $j\in\{2,\ldots,H\}$ satisfies
\[
 \sum_{i\in N_+}\mathbb E_{P\sim\widetilde D_i}[d_{j,iP}]
 \le\frac{\sum_{i\in N_+}b_i}{1-qc}|C_j|
 \le\frac{W}{1-qc}|C_j|.
\]
The coefficient of $|C_j|$ is at most
\[
 \frac{W}{1-qc}
 \le\frac{1-8/q}{1-1/q}
 =1-\frac7{q-1}\le1-\frac5q.
\]
Claim~\ref{cl:geometry} now bounds the expected request by the
available ordinary supply:
\[
 \sum_{i\in N_+}\mathbb E_{P\sim\widetilde D_i}[d_{j,iP}]
 \le(1-5/q)|C_j|
 \le(1-5/q)(1+2/q)|C_{j-1}|
 \le(1-3/q)|C_{j-1}|\le|S_{j-1}|,
\]
which concludes the proof of the claim.
\end{proof}

\subsection{Total  Fractional Entitlement }

The feasible set is a nonempty bounded set defined by finitely many
linear constraints. Choose a vertex $x=(x_{iP})$ of this set, meaning
a feasible point that is not the midpoint of two distinct feasible
points. Define the integral agents and fractional agents,
respectively, by
\[
 I=\{i\in N_+:x_{iP}=1\text{ for some }P\in\mathcal P_i\},
 \qquad F=N_+\setminus I.
\]
Every integral agent has exactly one positive configuration weight.
Every fractional agent has at least two positive configuration
weights. For every real threshold $z\in(0,c]$, define
\[
 F_z=\{i\in F:b_i\ge z\},
 \qquad
 h(z)=1+\left\lceil\frac{\log(1/z)}{\log(1+1/q)}\right\rceil.
\]
The integer $h(z)$ is chosen so that the prescribed bucket size
at index $h(z)$ is at least $q/z$.

\begin{claim}[Counting fractional agents]\label{cl:aps-count}
For every threshold $z\in(0,c]$, the number of fractional agents
with entitlement at least $z$ satisfies
\[
 |F_z|\le h(z)\le2+2q\log(1/z).
\]
\end{claim}

\begin{proof}
Fix $z\in(0,c]$. Every agent $i\in F_z$ has cutoff
$h_i\le h(z)$. Indeed, if $h(z)<H$, bucket $C_{h(z)}$ is full
and has size at least $q/z\ge q/b_i$; if $h(z)\ge H$, the
inequality follows from $h_i\le H$.

For each agent $i\in F_z$, choose two distinct configurations
$P_i^+,P_i^-\in\mathcal P_i$ with positive weights. Let
$\delta_i$ be a real perturbation that we add to $x_{iP_i^+}$ and
subtract from $x_{iP_i^-}$. This preserves every agent's total
weight. To preserve the capacity rows up to $h(z)$, impose
\[
 \sum_{i\in F_z}
 (d_{j,iP_i^+}-d_{j,iP_i^-})\delta_i=0
 \qquad(2\le j\le\min\{h(z),H\}).
\]
There are at most $h(z)-1$ equations here. In every later row
$j>h(z)$, all these agents use tail requests, so the change in
the row is
\[
 |C_j|\sum_{i\in F_z}
 (\lambda_{iP_i^+}-\lambda_{iP_i^-})\delta_i.
\]
Consequently one further equation,
\[
 \sum_{i\in F_z}
 (\lambda_{iP_i^+}-\lambda_{iP_i^-})\delta_i=0,
\]
preserves all later rows. We may impose it even when no later
rows exist.

If $|F_z|>h(z)$, this homogeneous system has more unknowns than
equations, and therefore has a nonzero solution. Scale that solution
so that
\[
 |\delta_i|\le\tfrac12\min\{x_{iP_i^+},x_{iP_i^-}\}
 \qquad(i\in F_z).
\]
Let $\Delta$ be the vector whose coordinates are
$\Delta_{iP_i^+}=\delta_i$, $\Delta_{iP_i^-}=-\delta_i$, and zero
elsewhere. Both $x+\Delta$ and $x-\Delta$ are distinct feasible
points with midpoint $x$, contradicting the choice of a vertex.
This proves $|F_z|\le h(z)$. Finally,
$\log(1+1/q)\ge1/(2q)$ gives
\[
 h(z)\le2+\frac{\log(1/z)}{\log(1+1/q)}
 \le2+2q\log(1/z),
\]
which concludes the proof of the claim.
\end{proof}

\begin{claim}[Bounding the total entitlement of fractional agents]\label{cl:aps-mass}
The fractional agents have total entitlement at most $1/(2q)$, i.e.,
\[
 \sum_{i\in F}b_i\le\frac1{2q}.
\]
\end{claim}

\begin{proof}
An agent $i\in F$ belongs to $F_z$ precisely when $0<z\le b_i$.
Integrating the bound in Claim~\ref{cl:aps-count} yields
\[
 \sum_{i\in F}b_i
 =\int_0^c|F_z|\,dz
 \le2c+2q\int_0^c\log(1/z)\,dz
 =2c+2qc\bigl(\log(1/c)+1\bigr).
\]
Our choice of $c$ and the inequality $q\ge16$ imply
\[
 \log(1/c)+1=\log32+2\log q+\log\log q+1
 \le6\log q.
\]
Substituting this bound gives
\[
 \sum_{i\in F}b_i
 \le2c+12qc\log q
 \le16qc\log q=\frac1{2q},
\]
which concludes the proof of the claim.
\end{proof}

\subsection{Constructing the Allocation}

For every integral agent $i\in I$, let $P_i\in\mathcal P_i$
denote her unique configuration with $x_{iP_i}=1$. Because all
request coefficients are nonnegative, dropping the fractional
agents' terms from \eqref{eq:aps-capacity} gives
\[
 \sum_{i\in I}d_{j,iP_i}\le|S_{j-1}|
 \qquad(2\le j\le H).
\]
Lemma~\ref{lem:realize} therefore supplies, in any realized outcome, pairwise disjoint bundles in $S\cup E$ giving every
$i\in I$ normalized value at least $u_i(P_i)\ge1$.

\begin{claim}[Serving the fractional agents]\label{cl:aps-reserve}
The reserve items $R$ admit pairwise disjoint bundles giving every
fractional agent $i\in F$ normalized value at least one.
\end{claim}

\begin{proof}
For each reserved rank $r\in R$, the consecutive ranks from $r$
to $\min\{r+q-1,m\}$ form a block of at most $q$ goods.
These blocks partition $M$, and their first ranks are most valuable
for every ordered valuation. For each active agent $i\in N_+$,
the certificate gives $u_i(M)\ge1/b_i$, so
\[
 \frac1{b_i}\le u_i(M)
 =\sum_{r\in R}\sum_{g=r}^{\min\{r+q-1,m\}}u_i(g)
 \le q\sum_{r\in R}u_i(r)=q\,u_i(R).
\]
For every fractional agent $i\in F$ and good $g\in R$, define
the fractional reserve share $z_{ig}=2qb_i$. By
Claim~\ref{cl:aps-mass}, the total fraction assigned of any
reserve good $g\in R$ is at most one:
\[
 \sum_{i\in F}z_{ig}=2q\sum_{i\in F}b_i\le1.
\]
Each fractional agent $i\in F$ receives fractional value at
least two:
\[
 \sum_{g\in R}z_{ig}u_i(g)=2qb_i u_i(R)\ge2.
\]
The normalized singleton values satisfy $u_i(g)\le1$.
Apply Lemma~\ref{lem:round} to this fractional allocation with
loss bound $m_i=1$ for each $i\in F$. In any outcome, the
resulting reserve bundles $A_i\subseteq R$, for $i\in F$,
are pairwise disjoint and satisfy
\[
 u_i(A_i)\ge\sum_{g\in R}z_{ig}u_i(g)-1\ge1,
\]
which concludes the proof of the claim.
\end{proof}

\begin{proof}[Proof of Theorem~\ref{thm:aps}]
Assume for simplicity that $1/\varepsilon$ is an integer.
The bundles for $I$ use $S\cup E$, and the bundles for $F$ use
the disjoint reserve $R$. Together they give normalized value at
least one, and hence original value at least $t_i$, to every
active agent $i\in N_+$. Give each inactive agent an empty bundle
and assign all unallocated goods arbitrarily. The ordering reduction
from Section~\ref{sec:prelim} transfers this guarantee to the original
instance. This proves Theorem~\ref{thm:aps}.
\end{proof}
\section{Best-of-Both-Worlds Randomized Allocations}
\label{sec:bobw}

We first prove a result for instances with total entitlement strictly
below one. We condition each active agent's certificate on configurations
with small total absolute contribution to the sums used in
Lemma~\ref{lem:doerr-distribution-preserving}. We then use that lemma
to select configurations jointly, preserving their distributions while
ensuring that all requests fit within the available supplies.
Lemma~\ref{lem:realize} realizes the selected configurations, and
we assign unused goods to obtain the required expected values.
Our best-of-both-worlds results (Theorems~\ref{thm:bobw} and \ref{thm:bobw2}) are implied  by  omitting a random group of agents or by
reducing all entitlements, respectively.

Fix $0<\delta\le1/2$ with $1/\delta$ an integer, and define
\[
 \kappa_\delta=\frac{\delta^3}{32768\log(8/\delta)}.
\]

\begin{proposition}[Proportionality and APS with entitlement slack]
\label{prop:bobw-slack}
Consider an instance with nonnegative additive valuations and
entitlements satisfying
\[
 0<W=\sum_{i\in N}b_i\le1-\delta,
 \qquad 0<b_i\le\kappa_\delta\quad(i\in N).
\]
There is a randomized complete allocation $A$ such that, for every
$i\in N$,
\[
 \E[v_i(A_i)]\ge\frac{b_i}{W}v_i(M),
 \qquad v_i(A_i)\ge\APS_i(b_i)\quad\text{in every outcome}.
\]
\end{proposition}

\begin{proof}
Write $t_i=\APS_i(b_i)$ and $N_+=\{i\in N:t_i>0\}$.
If $N_+=\varnothing$, assign every good to agent $i$ with probability
$b_i/W$; all APS targets are zero and the expected values are as
required. Otherwise, work with the ordered instance from
Section~\ref{sec:prelim} and use the construction of
Section~\ref{sec:ingredients} with
\[
 q=\frac8\delta,\qquad c=\kappa_\delta=\frac1{64q^3\log q}.
\]
Thus $q\ge16$ is an integer, $W\le1-8/q$, and
$b_i\le c<1/q^2$. We use the certificates, capped valuations,
buckets, supplies, cutoffs, tail proportions, and requests defined
in that section.

\paragraph{The available supplies.}
Claim~\ref{cl:geometry} implies
\begin{equation}\label{eq:bw-supply}
 |S_{j-1}|\ge\frac{1-3/q}{1+2/q}|C_j|
 \ge(1-5/q)|C_j|
 \qquad(2\le j\le H).
\end{equation}
The first inequality combines the ordinary-supply and bucket-growth
bounds. The second uses
$(1-5/q)(1+2/q)=1-3/q-10/q^2\le1-3/q$.

\paragraph{Conditioning the certificates.}
For every $i\in N_+$ and $P\subseteq M$, define
\[
 K_i(P)=\sum_{j=1}^{h_i}\frac{|P\cap C_j|}{|C_j|}
        +h_i\lambda_{iP}.
\]
Since $|C_1|=q$, the condition $K_i(P)\le1/(2q)$ implies
$P\cap C_1=\varnothing$: a single good from $C_1$ would contribute
$1/q$ to $K_i(P)$. Moreover,
\[
 \E_{P\sim D_i}[K_i(P)]\le h_i b_i+h_i b_i=2h_i b_i.
\]
Indeed, the certificate's inclusion bounds give expectation at most
$b_i$ for each bucket term, and Claim~\ref{cl:request-means} gives
$\E_{P\sim D_i}[\lambda_{iP}]\le b_i$.

Let $\alpha_i=\Pr_{P\sim D_i}[K_i(P)\le1/(2q)]$. Then
\begin{align*}
 1-\alpha_i
 &\le2q\E_{P\sim D_i}[K_i(P)]\le4qh_i b_i\\
 &\le8qc+8q^2c\log(1/c)
 \le49q^2c\log q=\frac{49}{64q}<\frac1q.
\end{align*}
The first inequality is Markov's inequality, and the second uses the
preceding expectation bound. The third uses Claim~\ref{cl:cutoff},
$b_i\le c$, and the monotonicity of $b\log(1/b)$ on $(0,c]$.
The fourth uses $\log(1/c)\le6\log q$ and $q\log q\ge8$, both
valid for $q\ge16$; the equality substitutes the definition of $c$.

Let $\widetilde D_i$ be $D_i$ conditioned on $K_i(P)\le1/(2q)$,
and let $\mathcal P_i$ be its support. Every $P\in\mathcal P_i$
avoids $C_1$ and satisfies $u_i(P)\ge1$ by Claim~\ref{cl:capping},
so it is a configuration. Each good has inclusion probability at most
$b_i/\alpha_i\le b_i/(1-1/q)$ under $\widetilde D_i$.
Claim~\ref{cl:request-means} therefore yields
\begin{equation}\label{eq:bw-request-means}
 \E_{P\sim\widetilde D_i}[d_{j,iP}]
 \le\frac{b_i|C_j|}{1-1/q}
 \qquad(i\in N_+,\ 2\le j\le H).
\end{equation}

\paragraph{Selecting configurations jointly.}
For each $i\in N_+$, define $H$ real-valued functions on
$\mathcal P_i$ by
\[
 f_{1,i}(P)=\lambda_{iP},\qquad
 f_{j,i}(P)=
 \begin{cases}
 |P\cap C_j|/|C_j|-\lambda_{iP},&j\le h_i,\\
 0,&j>h_i
 \end{cases}
 \quad(2\le j\le H).
\]
They satisfy
\begin{equation}\label{eq:bw-request-decomposition}
 f_{1,i}(P)+f_{j,i}(P)=\frac{d_{j,iP}}{|C_j|}
 \qquad(2\le j\le H).
\end{equation}
For $j\le h_i$, the two tail terms cancel. For $j>h_i$, the second
term is zero and $d_{j,iP}=|C_j|\lambda_{iP}$. Moreover,
\[
 \sum_{r=1}^H|f_{r,i}(P)|
 \le\sum_{j=2}^{h_i}\frac{|P\cap C_j|}{|C_j|}
       +h_i\lambda_{iP}
 =K_i(P)\le\frac1{2q}.
\]
The first inequality uses the triangle inequality and
$\lambda_{iP}\ge0$. The equality uses $P\cap C_1=\varnothing$,
and the last inequality holds throughout $\mathcal P_i$.

Apply Lemma~\ref{lem:doerr-distribution-preserving} with $I=N_+$,
$Q_i=\widetilde D_i$, the $H$ functions above, and $\xi=1/(2q)$.
It gives jointly sampled configurations $P_i\sim\widetilde D_i$
such that, in every outcome,
\[
 \left|\sum_{i\in N_+}f_{r,i}(P_i)
       -\sum_{i\in N_+}
          \E_{P\sim\widetilde D_i}[f_{r,i}(P)]\right|
 \le\frac1q
 \qquad(1\le r\le H).
\]
Adding the upper bounds for $r=1$ and $r=j$, and using
\eqref{eq:bw-request-decomposition}, gives
\[
 \frac1{|C_j|}\sum_{i\in N_+}d_{j,iP_i}
 \le\frac1{|C_j|}\sum_{i\in N_+}
          \E_{P\sim\widetilde D_i}[d_{j,iP}]+\frac2q
 \le\frac{W}{1-1/q}+\frac2q
 \le1-\frac5q.
\]
The second inequality uses \eqref{eq:bw-request-means} and
$\sum_{i\in N_+}b_i\le W$. The third uses $W\le1-8/q$ and
$(1-8/q)/(1-1/q)+2/q=1-7/(q-1)+2/q\le1-5/q$.
Combining this with \eqref{eq:bw-supply} gives
\[
 \sum_{i\in N_+}d_{j,iP_i}\le|S_{j-1}|
 \qquad(2\le j\le H).
\]

\paragraph{Realizing the configurations.}
Apply Lemma~\ref{lem:realize} to the selected configurations in each
outcome. Let $B_i$ be its partial bundles for $i\in N_+$, and set
$B_i=\varnothing$ for $i\notin N_+$. Every active agent satisfies
$u_i(B_i)\ge u_i(P_i)\ge1$ in every outcome. Define
\[
 \mu_{ig}=\Pr(g\in B_i)\qquad(i\in N,\ g\in M),
\]
including both configuration selection and realization in the
probability. We show that $\mu_{ig}\le b_i/W$ for every $i,g$.

Fix $i\in N_+$. For an ordinary good $g\in S_{j-1}$ with
$2\le j\le H$,
\[
 \mu_{ig}
 =\frac{\E_{P\sim\widetilde D_i}[d_{j,iP}]}{|S_{j-1}|}
 \le\frac{b_i}{(1-1/q)(1-5/q)}
 \le\frac{b_i}{1-6/q}\le\frac{b_i}{W}.
\]
The equality uses Lemma~\ref{lem:realize} and the preservation of
$\widetilde D_i$ during configuration selection. The first inequality
uses \eqref{eq:bw-request-means} and \eqref{eq:bw-supply}.
The second uses $(1-1/q)(1-5/q)\ge1-6/q$, and the third uses
$W\le1-8/q$.

For a buffer good $g\in E_{h_i}$ with $h_i<H$,
\[
 \mu_{ig}=\frac1{|E_{h_i}|}
 \le\frac{b_i}{1-b_i}\le\frac{b_i}{W}.
\]
The equality follows from Lemma~\ref{lem:realize}. The cutoff
definition gives $|C_{h_i}|\ge q/b_i$, so
$|E_{h_i}|=\lfloor|C_{h_i}|/q\rfloor\ge1/b_i-1$, proving
the first inequality. The second uses $b_i\le c<8/q\le1-W$.
All remaining assignment probabilities, including those of agents
outside $N_+$, are zero. Thus $\mu_{ig}\le b_i/W$ for all $i,g$.

\paragraph{Completing the allocation.}
Whenever a good $g$ is unused, give it to agent $i$ with probability
\[
 \frac{b_i/W-\mu_{ig}}{1-\sum_{k\in N}\mu_{kg}}.
\]
The numerator is nonnegative, and these probabilities sum to one.
The denominator is the probability that $g$ is unused. If it is zero,
$g$ is never unused in an outcome of positive probability, and its
existing assignment probabilities already equal $b_i/W$ for every
agent $i$. Otherwise, the completed assignment probability is
\[
 \mu_{ig}
 +\left(1-\sum_{k\in N}\mu_{kg}\right)
    \frac{b_i/W-\mu_{ig}}{1-\sum_{k\in N}\mu_{kg}}
 =\frac{b_i}{W}.
\]
Let $A$ be the completed allocation. It adds goods to each $B_i$, so, in the ordered instance,
\begin{equation}\label{eq:bw-slack-guarantees}
 \E[v_i(A_i)]=\frac{b_i}{W}v_i(M),
 \qquad v_i(A_i)\ge t_i\quad\text{in every outcome}.
\end{equation}
The equality uses additivity. For $i\in N_+$, the inequality follows
from $v_i(g)\ge t_i u_i(g)$ and $u_i(B_i)\ge1$; it is automatic
when $t_i=0$.

Apply the ordered-instance
reduction from Section~\ref{sec:prelim} to each outcome and assign any remaining goods arbitrarily. The reduction preserves $v_i(M)$
and each agent's APS, and never decreases her value in any outcome.
Thus \eqref{eq:bw-slack-guarantees} proves the proposition for the
original instance.
\end{proof}

\subsection{Proofs of the BOBW Results}
Fix $0<\varepsilon\le1/2$ and assume for simplicity that $1/\varepsilon$ is an integer. Let 
\[
 c'_\varepsilon=\frac{\varepsilon^3}{262144\log(16/\varepsilon)}
 =\kappa_{\varepsilon/2}.
\]
For Theorems~\ref{thm:bobw} and \ref{thm:bobw2}, consider an instance with nonnegative additive
valuations satisfying $W=\sum_{i\in N}b_i=1$ and
$0<b_i\le c'_\varepsilon$ for every $i\in N$.
We first prove Theorem~\ref{thm:bobw}.

\begin{proof}[Proof of Theorem~\ref{thm:bobw}]
Let $L=1+1/\varepsilon$. Partition $N$ into groups $G_1,\ldots,G_L$
by repeatedly placing an agent in a group of minimum current total
entitlement. The difference between any two group totals is at most
$c'_\varepsilon$: adding an agent to a smallest group increases its
total by at most $c'_\varepsilon$ and preserves this bound. Write
\[
 W_\ell=\sum_{i\in N\setminus G_\ell}b_i
 \qquad(1\le\ell\le L).
\]
Then
\[
 1-W_\ell\ge\frac1L-c'_\varepsilon
 \ge\frac{\varepsilon}{1+\varepsilon}-\frac\varepsilon6
 \ge\frac\varepsilon2.
\]
The first inequality uses the bound on the group totals and their
average $1/L$. The second uses $c'_\varepsilon\le\varepsilon/6$, and
the third uses $\varepsilon\le1/2$. All groups are nonempty, so
$0<W_\ell\le1-\varepsilon/2$.

For each $\ell$, apply Proposition~\ref{prop:bobw-slack} with
$\delta=\varepsilon/2$ to $N\setminus G_\ell$, keeping all goods
and the original entitlements. Its hypotheses hold because
$b_i\le c'_\varepsilon=\kappa_{\varepsilon/2}$. Denote its allocation
by $A^\ell=(A_i^\ell)_{i\in N}$, setting
$A_i^\ell=\varnothing$ for $i\in G_\ell$.

Choose a random group index $Z\in\{1,\ldots,L\}$ with
\[
 p_\ell=\Pr(Z=\ell)=\frac{W_\ell}{L-1}.
\]
These probabilities sum to one because $\sum_{\ell=1}^L W_\ell=L-1$,
and each satisfies $p_\ell\le1/(L-1)=\varepsilon$.
Conditional on $Z=\ell$, sample $A^\ell$ and set $A=A^\ell$.

Fix $i\in G_k$. She receives her APS whenever $Z\ne k$, so her
failure probability is at most $p_k\le\varepsilon$. Moreover,
\[
 \E[v_i(A_i)]
 =\sum_{\substack{\ell=1\\\ell\ne k}}^L
        p_\ell\E[v_i(A_i^\ell)]
 \ge\sum_{\substack{\ell=1\\\ell\ne k}}^L
        \frac{W_\ell}{L-1}\frac{b_i}{W_\ell}v_i(M)
 =b_i v_i(M).
\]
The first equality conditions on $Z$ and uses $A_i^k=\varnothing$.
The inequality uses Proposition~\ref{prop:bobw-slack} and the
definition of $p_\ell$; the final equality counts the $L-1$
indices different from $k$.
\end{proof}

The probability guarantee is individual: every agent succeeds with
probability at least $1-\varepsilon$. It does not assert that all agents
succeed simultaneously with that probability.

\begin{proof}[Proof of Theorem~\ref{thm:bobw2}]
Set $\widehat b_i=(1-\varepsilon)b_i$ and
$\widehat W=\sum_{i\in N}\widehat b_i=1-\varepsilon$.
Apply Proposition~\ref{prop:bobw-slack} with $\delta=\varepsilon/2$
to these entitlements, keeping the goods and valuations unchanged.
Its hypotheses hold because $\widehat W\le1-\varepsilon/2$ and
$\widehat b_i\le c'_\varepsilon=\kappa_{\varepsilon/2}$.
The resulting allocation satisfies
$v_i(A_i)\ge\APS_i((1-\varepsilon)b_i)$ in every outcome and
\[
 \E[v_i(A_i)]\ge\frac{\widehat b_i}{\widehat W}v_i(M)
 =\frac{(1-\varepsilon)b_i}{1-\varepsilon}v_i(M)
 =b_i v_i(M).
\]
The inequality is the proposition's expected-value guarantee, and
the first equality substitutes the modified entitlements and their sum.
\end{proof}

\section{Linear Hardness}
\label{sec:hardness}

\subsection{The Thirteen-Goods Instance}

Let $G=\{1,\ldots,13\}$ be a set of goods, and let $v$ be the
nonnegative additive valuation defined by
\[
 (v(1),\ldots,v(13))
 =(137,135,119,79,77,76,54,50,49,10,9,8,6).
\]
This instance appears in Table~2 of \citep{martinovic2016proper}. We give an explicit certificate
and a direct nonexistence proof.

Let $\mathcal Q$ be the family of bundles listed in the following
table. For every bundle $P\in\mathcal Q$, define its weight $\tau_P$
to be the number in the last column.
\begin{center}
\begin{tabular}{@{}lrr@{}}
\toprule
Bundle $P$ & Value $v(P)$ & Weight $\tau_P$\\
\midrule
$\{4,5\}$ & $156$ & $1/2$\\
$\{4,6\}$ & $155$ & $1/2$\\
$\{5,6\}$ & $153$ & $1/2$\\
$\{2,7\}$ & $189$ & $1/5$\\
$\{3,7\}$ & $173$ & $1/5$\\
$\{3,8\}$ & $169$ & $2/5$\\
$\{3,9\}$ & $168$ & $2/5$\\
$\{7,8,9\}$ & $153$ & $3/5$\\
$\{2,10,11\}$ & $154$ & $1/5$\\
$\{2,10,12\}$ & $153$ & $1/5$\\
$\{1,11,12\}$ & $154$ & $2/5$\\
$\{1,10,13\}$ & $153$ & $3/5$\\
$\{2,11,12,13\}$ & $158$ & $2/5$\\
\bottomrule
\end{tabular}
\end{center}
For every good $g\in G$, define its \emph{load} $\ell_g$ as the total
weight of the bundles containing it:
\[
 \ell_g=\sum_{\substack{P\in\mathcal Q\\g\in P}}\tau_P.
\]

\begin{observation}[Fractional certificate]\label{obs:hardness-certificate}
Every bundle $P\in\mathcal Q$ has value at least $153$. The weights
have total $51/10$, and every good has load one:
\[
 \sum_{P\in\mathcal Q}\tau_P=\frac{51}{10},
 \qquad \ell_g=1\quad(g\in G).
\]
\end{observation}

\begin{claim}[No five target bundles]\label{cl:hardness-core}
There are no five pairwise disjoint bundles of value at least $153$.
\end{claim}

\begin{proof}
Suppose such bundles $B_1,\ldots,B_5\subseteq G$ exist, with
$v(B_j)\ge153$ for every index $j\in\{1,\ldots,5\}$.
Assign any unused goods to them, so that they partition $G$.
For every index $j\in\{1,\ldots,5\}$, define
the nonnegative excess $e_j=v(B_j)-153$. Since $v(G)=809$, their
total excess is
\[
 \sum_{j=1}^5e_j=809-5\cdot153=44.
\]
No bundle can contain two of goods $1,2,3$: even the two smallest of
these goods would give excess at least
\[
 135+119-153=101>44.
\]
Relabel the bundles so that $1\in B_1$, $2\in B_2$, and $3\in B_3$.

Define $L=\{10,11,12,13\}$ to be the set of the four smallest goods.
Their total value is $v(L)=10+9+8+6=33$. Bundle $B_3$ needs at least
$153-119=34$ additional value, so it must contain a good from
$\{4,\ldots,9\}$. Every such good is worth at least $49$, giving
\[
 e_3\ge119+49-153=15.
\]
Together, the disjoint bundles $B_1$ and $B_2$ need additional value
at least
\[
 (153-137)+(153-135)=34>v(L).
\]
Thus one of them, denoted $B_j$ for an index $j\in\{1,2\}$, also
contains a good from $\{4,\ldots,9\}$. Its excess satisfies
\[
 e_j\ge135+49-153=31.
\]
The two distinct bundles $B_3$ and $B_j$ therefore have combined
excess at least $15+31=46$, contradicting the total excess of $44$.
\end{proof}

\subsection{Adding Agents and Singleton Goods}

\begin{proposition}[A family with arbitrarily small entitlements]
\label{prop:hardness-padding}
For every integer $n\ge5$, define
\[
 \beta_n=\frac{10}{10n+1}.
\]
There is an instance with agent set $N=\{1,\ldots,n\}$, identical
additive valuations, and entitlements $b_i=\beta_n$ for every agent
$i\in N$, such that
\[
 \sum_{i\in N}b_i=1-\frac1{10n+1},
 \qquad \APS_i(\beta_n)\ge153\quad(i\in N),
\]
but no allocation gives every agent value at least $153$.
\end{proposition}

\begin{proof}
Let $G_{\mathrm{new}}$ be a set of $n-5$ additional goods, disjoint
from $G$, and set $M=G\cup G_{\mathrm{new}}$. Extend $v$ by assigning
value $153$ to every good in $G_{\mathrm{new}}$. Every agent has
valuation $v_i=v$ and entitlement $b_i=\beta_n$.

Define the extended bundle family
\[
 \mathcal Q_n=\mathcal Q\cup\{\{g\}:g\in G_{\mathrm{new}}\}.
\]
Retain the weights $\tau_P$ for the bundles $P\in\mathcal Q$ and
define $\tau_{\{g\}}=1$ for every good $g\in G_{\mathrm{new}}$.
Observation~\ref{obs:hardness-certificate} gives the total weight and the
load of each good:
\[
 \sum_{P\in\mathcal Q_n}\tau_P
 =\frac{51}{10}+n-5=n+\frac1{10},
 \qquad
 \sum_{\substack{P\in\mathcal Q_n\\g\in P}}\tau_P=1
 \quad(g\in M).
\]
Define a distribution $D_n$ on subsets of $M$ by
\[
 D_n(P)=
 \begin{cases}
  \beta_n\tau_P,&P\in\mathcal Q_n,\\
  0,&P\notin\mathcal Q_n.
 \end{cases}
\]
Its probabilities sum to $\beta_n(n+1/10)=1$. Every bundle in its
support has value at least $153$, and each good $g\in M$ has
inclusion probability
\[
 \sum_{\substack{P\subseteq M\\g\in P}}D_n(P)
 =\beta_n\sum_{\substack{P\in\mathcal Q_n\\g\in P}}\tau_P
 =\beta_n.
\]
Hence $D_n$ is an APS certificate for target $153$ at entitlement
$\beta_n$. Lemma~\ref{lem:certificate} implies
$\APS_i(\beta_n)\ge153$ for every agent $i\in N$.

For any allocation $(A_i)_{i\in N}$, define
\[
 J=\{i\in N:A_i\cap G_{\mathrm{new}}\ne\varnothing\}
\]
to be the set of agents receiving a new good. Disjointness gives
$|J|\le|G_{\mathrm{new}}|=n-5$, so at least five agents receive only
goods from $G$. If all agents received value at least $153$, these
five bundles would contradict Claim~\ref{cl:hardness-core}.
Finally, the entitlement sum is
$n\beta_n=10n/(10n+1)=1-1/(10n+1)$, as claimed.
\end{proof}

\begin{proof}[Proof of Theorem~\ref{thm:hardness-cap}]
Fix $0<\varepsilon<1/51$, and define
\[
 n_\varepsilon=\left\lceil\frac{1/\varepsilon-1}{10}\right\rceil-1.
\]
This is the largest integer strictly below $(1/\varepsilon-1)/10$;
in particular, $n_\varepsilon\ge5$. Its defining inequalities give
\[
 \frac1\varepsilon-10\le10n_\varepsilon+1<\frac1\varepsilon.
\]
Apply Proposition~\ref{prop:hardness-padding} with $n=n_\varepsilon$.
The resulting instance has no APS allocation, and its entitlement
sum is strictly below $1-\varepsilon$:
\[
 \sum_{i\in N}b_i=1-\frac1{10n_\varepsilon+1}<1-\varepsilon.
\]
Since all its entitlements equal $\beta_{n_\varepsilon}$, no cap
$c\ge\beta_{n_\varepsilon}$ has the universal property defining
$c^\star(\varepsilon)$. Moreover, $1/\varepsilon-10>0$, and $1-10\varepsilon\geq 10/13$ so
\[
 c^\star(\varepsilon)\le\beta_{n_\varepsilon}
 =\frac{10}{10n_\varepsilon+1}
 \le\frac{10\varepsilon}{1-10\varepsilon} \leq 13\varepsilon,
\]
which concludes the proof.
\end{proof}

\section{Discussion}
\subsection{Polynomial-Time Implementation}
\label{sec:implemetaion}
All our positive results admit polynomial-time implementations for the same order of bound on the maximum entitlement.  
Consider an instance with $W\le1-\delta$ and set
$b_i'=(1+\delta/2)b_i$. Knapsack price rounding and multiplicative
weights compute values $x_i$ satisfying
\[
 \APS_i(b_i)\le x_i\le\APS_i(b_i'),
\]
together with polynomial-support certificates for $x_i$ at entitlement
$b_i'$.\footnote{At any price vector, rounding prices down allows
knapsack dynamic programming to find a bundle of value at least
$\APS_i(b_i)$ and price at most $(1+\delta/4)b_i$.
Multiplicative weights averages polynomially many such bundles to obtain
item marginals at most $(1+\delta/2)b_i$; take $x_i$ to be their minimum
value. Entitlements $b_i<1/m$ have APS zero and need no computation.
For rational inputs, the running time is polynomial in the input size
and $1/\delta$.}
Since
\begin{equation}
\label{eq:bi}    
 \sum_i b_i'\le1-\delta/2,
 \qquad \frac{b_i'}{\sum_j b_j'}=\frac{b_i}{W},
\end{equation}
we run the existing constructions with entitlements $b_i'$, targets
$x_i$, and slack parameter $\delta/2$. The proofs require only the
certificates for their targets, and proportionality is preserved by
Equation~\eqref{eq:bi}. An entitlement cap $C_\delta$ is thereby replaced
by $C_{\delta/2}/(1+\delta/2)$, preserving the respective orders
$\varepsilon^2/\log(1/\varepsilon)$ and $\varepsilon^3/\log(1/\varepsilon)$.
All remaining allocation and rounding steps are polynomial;
exact randomized sampling takes polynomial expected time.

\subsection{MMS Allocations with Resource Augmentation}

\label{sec:resource-augmentation}

\citet{akrami2025fair} study a complementary approach to MMS fairness:
retain every agent's original MMS and allow additional copies of
selected goods. Each copy has the same value as the original good to
every agent, and all MMS benchmarks are computed before duplication.
For additive valuations, they show that $n-2$ distinct copies suffice,
and also obtain a bound of $(m/3)(1+o(1))$
\citep{akrami2025fair}.
They ask whether one extra copy always suffices and, more generally,
how far these bounds can be reduced
\citep{akrami2025fair}.
They further point out that, for additive valuations, a general guarantee
of MMS with fewer than $n/4$ copies would improve the then-best-known
$1$-out-of-$d$ MMS bound, so achieving such a guarantee may be
nontrivial~\citep{akrami2025fair}.
Combining our ordinal MMS guarantee with their reduction
\citep[Lemma~3.10]{akrami2025fair} yields the following improvement.

\begin{corollary}[MMS with few duplicated goods]
\label{cor:mms-duplicates}
For every instance with $n\ge2$ agents and $m$ goods, with nonnegative
additive valuations, adding one copy of each of at most
\[
 O\!\left(m\sqrt{\frac{\log n}{n}}\right)
\]
goods suffices to give every agent $i$ value at least
$\mathrm{MMS}_i(n)$, evaluated on the original goods.
\end{corollary}

The proportion of goods that must be duplicated therefore tends to
zero as $n$ grows. Together with the $n-2$ bound, this gives a bound of
$\min\{n-2,\,O(m\sqrt{\log n/n})\}$ distinct copies; in particular,
$O(\sqrt{n\log n})=o(n)$ copies suffice when $m=O(n)$.
In this regime, the number of copies is therefore below $n/4$ for
sufficiently large $n$.

\subsection{Beyond Additive Valuations}
Our results extend to valuations of the form
$v_i(S)=\phi_i(a_i(S))$, where $a_i$ is a nonnegative additive
valuation and $\phi_i$ is nondecreasing and concave with
$\phi_i(0)=0$. Indeed, APS and MMS under $v_i$ are obtained by
applying $\phi_i$ to the corresponding shares under $a_i$,
so the pointwise and probability guarantees transfer.
Moreover, concavity gives
\[
 v_i(S)\ge \frac{a_i(S)}{a_i(M)}v_i(M)
 \qquad\text{whenever }a_i(M)>0,
\]
so ex-ante proportionality also transfers; if $a_i(M)=0$,
all values are zero. This class includes budget-additive
valuations, for which $\phi_i(x)=\min\{x,B_i\}$.
For the deterministic APS existence result, monotonicity
alone suffices.

\subsection{Limitations of the Bidding Game Approach}
\label{subsec:method-obstructions}

\paragraph{The bidding game.}
A common approach for proving the existence of fair allocations is through the bidding game~\citep{DBLP:conf/sigecom/KalaiMT15,babaioff2024fair,DBLP:journals/corr/abs-2303-12444,babaioff2025share,DBLP:conf/sigecom/FeigeG25}. In a bidding game, the agents have budgets equal to
their entitlements; the highest bidder buys any affordable nonempty
set, paying her bid per good. 
This approach establishes fair allocations by designing safe strategies that guarantee each agent a target value against arbitrary opponent strategies.
Fix a distinguished agent $*$.
We show that she cannot guarantee her APS against arbitrary opponents,
even with fixed entitlement slack and arbitrarily small entitlements. In other words, 
our APS existence result cannot be established by giving
each agent a strategy that guarantees her APS against arbitrary
strategies of the other agents.
We next present why Theorem~\ref{thm:aps} cannot be obtained by the bidding game approach. Similar arguments also hold for Theorems~\ref{thm:bobw} and \ref{thm:bobw2}.

Let $k\ge100$ be even. All agents value each of $k$ large goods at $10$
and each of $2k$ small goods at $1$. Agent $*$ has entitlement $1/k$;
her $k$ opponents each have entitlement $15/(16k)$. Hence
\[
 W=\frac{15}{16}+\frac1k<\frac{19}{20},
 \qquad \max_i b_i=\frac1k,
 \qquad \APS_*(1/k)=12.
\]
The uniform distribution over $k$ disjoint triples, each containing one
large and two small goods, certifies value $12$; prices proportional
to values give the matching upper bound.

Measure money in units of $1/(32k)$: agent $*$ has budget $32$ and
every opponent has budget $30$. Until $*$ first wins, opponents act
one at a time: the first $k/2$ bid $21$, then the others bid $27$, each
taking one large good. All other opponents bid zero.
If $*$ wins during these rounds, she buys only one good, worth at most
$10$, since her winning bid is at least $21$.
We consider three mutually exclusive and exhaustive cases.

\smallskip\noindent\textbf{Case 1: First win in the $21$ phase.}
Agent $*$ has at most $11$ left. Previously unspent opponents remove
the remaining large goods at bid $12$, one each. They suffice even if
$*$ bought a small good, since earlier each opponent took one large
good. Agent $*$ cannot match this bid. At least half the opponents
retain at least $18$, and all others retain at least $9$.
At bid $6$, half can therefore buy three small goods each and the rest
one each, enough for all $2k$ small goods. Agent $*$ can afford at most
one additional small good, so her total value is at most $11$.

\smallskip\noindent\textbf{Case 2: First win in the $27$ phase.}
Agent $*$ has at most $5$ left. As in Case 1, previously unspent opponents remove
the remaining large goods at bid $6$, one each, which $*$ cannot match.
The first half of the opponents retain $9$ each, and all others retain
at least $3$. At bid $3$, half can buy three small goods each and the
rest one each. Again, agent $*$ can afford at most one additional
small good and receives total value at most $11$.

\smallskip\noindent\textbf{Case 3: No win before all large goods are removed.}
The two halves retain budgets $9$ and $3$, respectively. At bid $3$,
they can buy three and one small goods each, respectively, exhausting
the supply. Agent $*$ has budget $32$, so she can buy at most ten
small goods.

In every case, opponents act one at a time, repeating after losses,
until their prescribed purchases are complete or no goods remain.
All opposing bids are integers in the chosen units; agent $*$ may
use arbitrary real bids, and the argument holds under every tie-breaking rule.

The strategy uses only past events, so the bound also holds against
randomized strategies. 
Thus, individual guarantees against arbitrary opponents cannot yield our
results.

\section*{Acknowledgments}

\paragraph{Use of AI tools.}
We used OpenAI's ChatGPT and Codex to assist with exposition,
literature searches, and the exploration of mathematical arguments.
The authors take full responsibility for the mathematical arguments,
references, and conclusions presented in this paper.

\bibliographystyle{abbrvnat}
\bibliography{bib}
\end{document}